\documentclass[sigconf]{acmart}

\def\submission{0}
\usepackage{ising_v1}

\usepackage{xcolor}
\usepackage[ruled,linesnumbered]{algorithm2e}
\usepackage{multirow}
\usepackage{pifont}
\usepackage{subfigure}
\usepackage{cleveref}

\newcommand{\solution}{FedSVD\xspace}

\newtheorem{theorem}{Theorem}

\newcommand{\tabincell}[2]{\begin{tabular}{@{}#1@{}}#2\end{tabular}}

\Crefname{equation}{Eq.}{Eqs.}
\Crefname{figure}{Fig.}{Figs.}
\Crefname{tabular}{Tab.}{Tabs.}
\Crefname{table}{Tab.}{Tabs.}

\AtBeginDocument{%
  }

\copyrightyear{2022}
\acmYear{2022}
\setcopyright{acmcopyright}
\acmConference[KDD '22] {Proceedings of the 28th ACM SIGKDD Conference on Knowledge Discovery and Data Mining}{August 14--18, 2022}{Washington, DC, USA.}
\acmBooktitle{Proceedings of the 28th ACM SIGKDD Conference on Knowledge Discovery and Data Mining (KDD '22), August 14--18, 2022, Washington, DC, USA}
\acmPrice{15.00}
\acmDOI{10.1145/3534678.3539402}
\acmISBN{978-1-4503-9385-0/22/08}

\begin{document}

\sloppy

\title{Practical Lossless Federated Singular Vector Decomposition over Billion-Scale Data}


%

\settopmatter{authorsperrow=4}

\author{Di Chai}
\email{dchai@cse.ust.hk}
\affiliation{%
	\institution{Hong Kong University of Science and Technology\\Clustar Co., Ltd}
	\country{}
}
\author{Leye Wang}
\email{leyewang@pku.edu.cn}
\affiliation{
	\institution{MOE Key Lab of High Confidence Software Technologies,\\Peking University}
	\country{}
}

\author{Junxue Zhang}
\email{jzhangcs@cse.ust.hk}
\affiliation{
	\institution{Hong Kong University of Science and Technology\\Clustar Co., Ltd}
	\country{}
}

\author{Liu Yang}
\email{lyangau@cse.ust.hk}
\affiliation{
	\institution{Hong Kong University of Science and Technology\\Clustar Co., Ltd}
	\country{}
}
\author{Shuowei Cai}
\email{scaiak@cse.ust.hk}
\affiliation{
	\institution{Hong Kong University of Science and Technology\\Clustar Co., Ltd}
	\country{}
}

%

\author{Kai Chen}
\email{kaichen@cse.ust.hk}
\affiliation{
	\institution{Hong Kong University of Science and Technology}
	\country{}
}

\author{Qiang Yang}
\email{qyang@cse.ust.hk}
\affiliation{
	\institution{Hong Kong University of Science and Technology\\AI Group, WeBank Co., Ltd}
	\country{}
}

\renewcommand{\shortauthors}{Di Chai et al.}

\begin{abstract}
	
With the enactment of privacy-preserving regulations, e.g., GDPR, federated SVD is proposed to enable SVD-based applications over different data sources without revealing the original data. However, many SVD-based applications cannot be well supported by existing federated SVD solutions. The crux is that these solutions, adopting either differential privacy (DP) or homomorphic encryption (HE), suffer from accuracy loss caused by unremovable noise or degraded efficiency due to inflated data.

In this paper, we propose \solution, a practical lossless federated SVD method over billion-scale data, which can simultaneously achieve lossless accuracy and high efficiency. At the heart of \solution is a lossless matrix masking scheme delicately designed for SVD: 1) While adopting the masks to protect private data, \solution completely removes them from the final results of SVD to achieve lossless accuracy; and 2) As the masks do not inflate the data, FedSVD avoids extra computation and communication overhead during the factorization to maintain high efficiency. Experiments with real-world datasets show that FedSVD is over $10000\times$ faster than the HE-based method and has 10 orders of magnitude smaller error than the DP-based solution ($\epsilon=0.1,\delta=0.1$) on SVD tasks. We further build and evaluate \solution over three real-world applications: principal components analysis (PCA), linear regression (LR), and latent semantic analysis (LSA), to show its superior performance in practice. On federated LR tasks, compared with two state-of-the-art solutions: FATE~\cite{liu2021fate} and SecureML~\cite{Secureml}, \solution-LR is $100\times$ faster than SecureML and $10\times$ faster than FATE.

\end{abstract}

\ccsdesc[500]{Security and privacy~Privacy-preserving protocols}
\ccsdesc[500]{Computing methodologies~Factorization methods}

\keywords{Federated Learning; SVD; Privacy-preserving Matrix Factorization}

\maketitle

\section{Introduction} \label{sec:introduction}

Singular vector decomposition~(SVD) is an essential primitive to build various data analytics and machine learning applications over large-scale data. SVD is widely used in 1) principal component analysis~(PCA) to reduce the dimensionality of large-scale features; 2) latent semantic analysis~(LSA) \cite{dumais2004latent} on large-scale natural language processing (NLP) tasks to extract compressed embedding features. \highlight{These large-scale data usually come from various data sources in the real world \cite{yang2019federated} and it is hard for a single institution to collect sufficient data to produce robust results.}

However, since the enacting of privacy-preserving laws, \eg, GDPR~\cite{gdpr}, the data from different sources are restricted from being collected in one central place for conventional centralized SVD computation. To solve the problem, pioneer researchers have explored the SVD in a \emph{federated}\footnote{We use the term of \emph{federated} since the Federated SVD works similarly as Federated Learning~\cite{yang2019federated}.} approach~(Federated SVD), \ie, the SVD computation can be performed cooperatively by different participants without revealing or gathering their original data. We will give a formal definition of Federated SVD in \S\ref{sec:federated_svd}.

In this paper, we show that these Federated SVD solutions~\cite{FedPCA,WDA,PPD-SVD,han2009privacy} cannot \emph{efficiently} and \emph{accurately} process large-scale data, making them impractical to support real-world SVD applications. Specifically, prior works either leverage the differential privacy~(DP) or the homomorphic encryption~(HE) for privacy protection. The DP-based solution~\cite{FedPCA} suffers from dramatic \emph{accuracy loss} because it adds unremovable noise in the data to hide individual privacy. Our evaluation results show that DP-based federated SVD ($\epsilon=0.1,\delta=0.1$) has over $10$ orders of magnitude larger accuracy loss compared with centralized SVD~(more details are presented in \S\ref{sec:motivation}). The inherent loss of data utility hindered its application in the real world \cite{sphinx}, \eg, inaccurate SVD results during \highlight{medical} studies \cite{DBLP:journals/jbi/LatifogluPKG08} can cause severe issues in the subsequent medical diagnosis tasks. In contrast, the HE-based solution~\cite{PPD-SVD} can achieve lossless accuracy. However, the HE involves large computation/communication \emph{overhead} due to the inflated data, causing significant performance degradation. Our evaluation results show that it takes $\sim 15$ years for HE-based methods~\cite{PPD-SVD} to factorize a 1K $\times$ 100K matrix, \ie, 100 million elements~(more details in in \S\ref{sec:motivation}).

We ask: \emph{Can we build a practical lossless and efficient federated SVD solution over billion-scale data}. Our answer is \solution. The core of \solution is a matrix masking method delicately designed for the SVD algorithm. The advantage of this matrix masking method is that it can simultaneously achieve lossless accuracy and high efficiency. Specifically, 1) The masking method protects users' private data by multiplying two random orthogonal matrices. These random masks can be removed entirely from the final SVD results to achieve \emph{lossless} accuracy; 2) Unlike HE-based methods that significantly inflate the original data (\eg, from 64-bits to 2048-bits), \solution's masking method does not inflate the data size. Therefore, \solution can achieve similar performance as centralized SVD theoretically. Furthermore, based on the matrix masking method, we design optimization strategies including block-based mask generation, efficient data masking/recovering through block matrix multiplication, mini-batch secure aggregation, and advanced disk offloading to further improve the efficiency of communication, computation, and memory usage~(more details in \S\ref{sec:design}). Eventually, we have provided privacy analysis and attack experiments showing that \solution is secure and the raw data cannot be revealed from the masked data when the hyper-parameter is appropriately settled.

We implement and evaluate \solution on SVD tasks and three applications: PCA, linear regression (LR), and LSA. Our evaluation results show that: 1) On SVD tasks, \solution has 10 orders of magnitude smaller error compared with DP-based methods, and \solution is more than $10000\times$ faster than HE-based solution. Approximately, the HE-based method needs more than 15 years to factorize 1K $\times$ 100K data (\ie, 100 million elements), while \solution only needs 16.3 hours to factorize 1K $\times$ 50M data which containers 50 billion elements; 2) On PCA application, \solution takes 32.3 hours to compute the top 5 principal components on 100K $\times$ 1M synthetic data, which contains 100 billion elements; 3) On LR application, we have compared \solution with two well-known federated LR solutions: SecureML \cite{Secureml} and FATE \cite{liu2021fate}. The evaluation results show that \solution is 100$\times$ faster than SecureML and 10$\times$ faster than FATE; 4) On LSA application, \solution takes 3.71 hours to compute the top 256 eigenvectors on a 62K $\times$ 162K MovieLens real-world datasets, which contains 10 billion elements; 5) We perform attack experiments showing that, given proper hyper-parameter, \solution is secure against state-of-the-art~(SOTA) ICA attack~\cite{ica_attack} which is delicately designed for masked data.

The \solution is fully open-sourced \footnote{\url{https://github.com/Di-Chai/FedEval/tree/master/research/FedSVD}} and we believe, besides the three mentioned applications, \solution can benefit more applications that require SVD as their cores under the increasingly strict data protection laws and regulations.

\section{Background \& Motivation}

\subsection{SVD and Federated SVD} \label{sec:federated_svd}

SVD decomposes matrix $\mathbf{X} \in \mathbb{R}^{m \times n}$ into a product of three matrices
\begin{equation}\label{eq:svd_def}
    \mathbf{X = U \Sigma V^T}
\end{equation}
where $\mathbf{U} \in \mathbb{R}^{m \times m}$ and $\mathbf{V^T} \in \mathbb{R}^{n \times n}$ are the left and right singular vectors, $\mathbf{\Sigma} \in \mathbb{R}^{m \times n}$ is a rectangular diagonal matrix containing the singular values. $\mathbf{U}$ and $\mathbf{V^T}$ are orthogonal matrices. 

SVD is an essential building block in many studies. Here we introduce two of the most well-known SVD-based applications and they all require lossless accuracy and large-scale performance to be simultaneously achieved. \highlight{1) Principal components analysis (PCA). PCA is one of the most essential techniques for eliminating redundancy in high-dimensional data, and it is widely used in medical diagnosis \cite{DBLP:journals/jbi/LatifogluPKG08}, biometrics \cite{price2006principal}, and many other applications \cite{sanguansat2012principal}. SVD is the standard solution to conduct PCA. The SVD-based PCA deals with large-scale private data containing high-dimensional features, and it also requires lossless accuracy to avoid severe issues like inaccurate disease analysis.} 2) Linear regression (LR). LR is a popular machine learning model commonly used for risk management, marketing, \etc, for its high efficiency and interpretability. SVD could serve as a basis of the least square solution to LR. Compared with stochastic gradient descent (SGD), solving LR through SVD requires only one iteration and guarantees that the result is the global optimum. In such commercial scenarios, SVD-based LR deals with large-scale sensitive user data and requires lossless precision to avoid financial loss.

Typically, the federated SVD is defined as following: assume we have $k$ parties, and each party i owns data matrix $\mathbf{X}_i \in \mathbb{R}^{m \times n_i}$. Those $k$ parties would like to carry out SVD jointly on data $\mathbf{X}=[\mathbf{X}_1, \mathbf{X}_2,...,\mathbf{X}_k]$, where $\mathbf{X} \in \mathbb{R}^{m \times n}$ and $n=\sum_{i=1}^k n_i$.

\vspace{-2mm}
\begin{equation}\label{eq:fedsvd_def}
    \mathbf{[X_1,...,X_i,...,X_k] = U \Sigma [V_1^T,...,V_i^T,...,V_k^T]}
\end{equation}
\vspace{-2.5mm}

\Cref{eq:fedsvd_def} shows the federated SVD results. Accordingly, in a federated SVD solution, the $i$-th party ($1 \le i \le k$) gets $\mathbf{X}_i = \mathbf{U} \mathbf{\Sigma} \mathbf{V}_i^T$, where $\mathbf{U, \Sigma}$ are shared results among all participants, and $\mathbf{V}_i^T\in \mathbb{R}^{n \times n_i}$ is the secret result possessed by party-$i$. Figure \ref{fig:problem} also illustrates the above problem definition. Party-$i$'s data (\textit{i.e.,} $\mathbf{X}_i$) cannot be leaked to any other parties during the computation.

\vspace{-4mm}
\begin{figure}[h]
	\centering
	\includegraphics[scale=0.3]{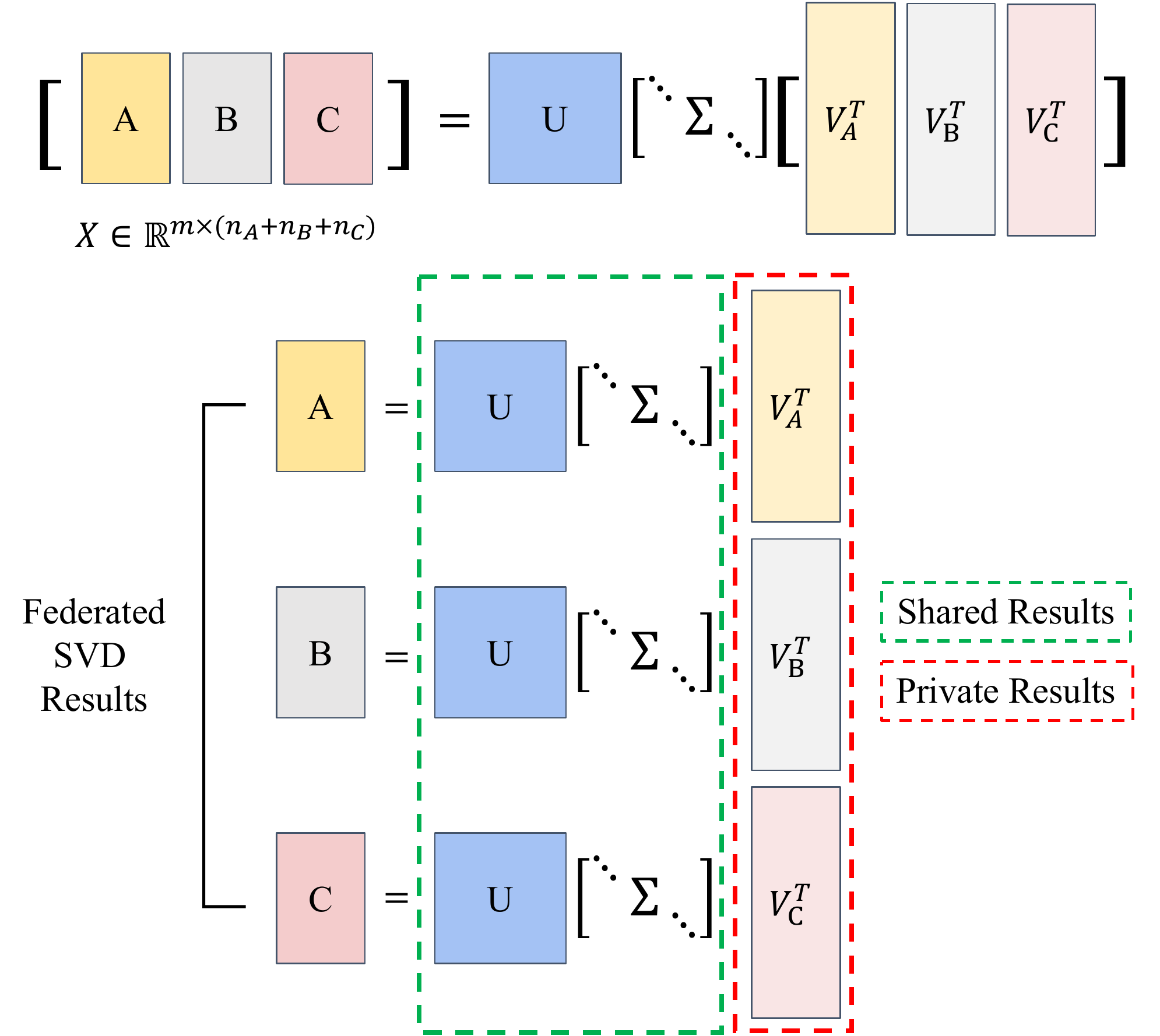}
	\caption{Problem formulation of federated SVD}
	\label{fig:problem}
\end{figure}

The real-world applications mainly contain two data partition scenarios,~\ie, horizontally and vertically partitioned scenarios. Horizontally partitioned scenario assumes that different parties share the same feature space but different sample space, while the vertically partitioned scenario assumes that participants share different feature space but the same sample space. In this paper, we do not make assumption on the data partition schema and our method is suitable for both two scenarios. Because one type of partition could be easily transferred to another through matrix transpose in SVD. Without loss of generality, as shown in \Cref{fig:problem}, we assume the data matrix is vertically partitioned among parties.

\subsection{Prior Work Suffering from Either Accuracy Loss or Performance Penalty}
\label{sec:motivation}

\vspace{-4.5mm}
\begin{figure}[h]
	\small
	\centering
	\subfigure[DP-SVD ($\delta=0.01$) has 7 $\sim$ 14 magnitudes larger error compared with \solution on four datasets.]{
		\label{fig:motivation_dp}
		\includegraphics[scale=0.38]{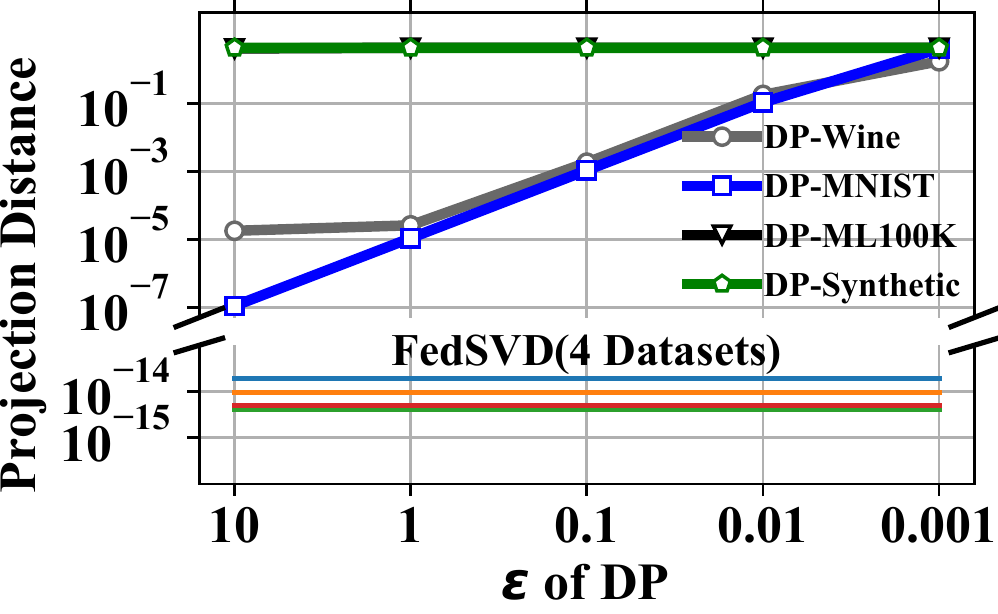}
	}
	\hfill
	\hfill
	\subfigure[HE-based SVD needs 15.1 years to factorize 1K $\times$ 100K data (100 million elements).]{
		\label{fig:motivation_he}
		\includegraphics[scale=0.38]{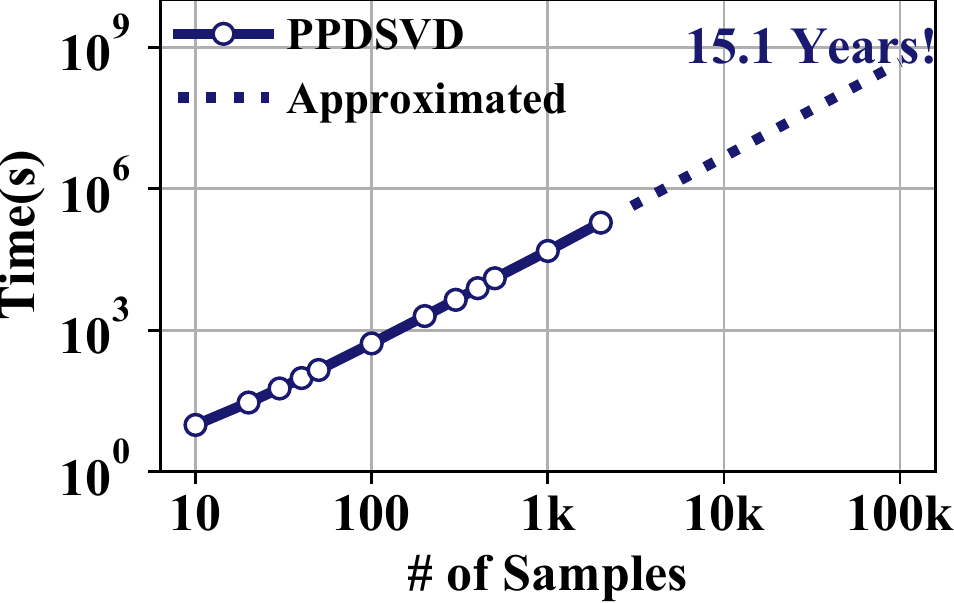}
	}
	\caption{Quantifying accuracy loss and performance penalty.}
	\label{fig:motivation}
\end{figure}

\vspace{1.5mm}

Pilot federated learning work designed privacy-preserving SVD methods in two brunches: the DP-based and HE-based methods.

\parab{Accuracy Loss:} On the one hand, \citet{FedPCA} proposed a federated and ($\epsilon,\delta$)-DP principal component analysis method, in which the leaf nodes apply DP locally and upload the local PCA results to one root node, which will asynchronously aggregate the received updates. Although DP-based solutions are easy to implement and does not have efficiency issue, it unavoidably brings loss to the data utility and hindered its application in real-world \cite{sphinx}. For example, accuracy loss of SVD in medical study can cause severe issues in subsequent medical diagnosis. \Cref{fig:motivation_dp} shows that DP-based SVD has 7 $\sim$ 14 orders of magnitude larger error compared with \solution under different parameters.

\parab{Performance Penalty:} On the other hand, \citet{PPD-SVD} proposed a HE-based SVD solution, in which the parties jointly compute the covariance under additive HE (\ie, HE algorithm that only supports addition operation on cipher-text), then a trusted server decrypts the covariance matrix and conducts the SVD. Although HE is lossless, it brings heavy computation and communication overhead because it swells up the data size from 64-bit to 2048-bit, assuming the key length is set to $2^{20}$ bits. Thus HE-based SVD has large computation overhead. In particular, as shown in \Cref{fig:motivation_he}, HE-based method needs more than 15 years to factorize a 1K $\times$ 100K data (\ie, million-scale data).

\parab{Conclusion:} None of the exiting federated SVD work can simultaneously achieve lossless accuracy and high efficiency.

\section{\solution}
\label{sec:design}

\begin{figure*}[h!]
	\centering
	\includegraphics[scale=0.48]{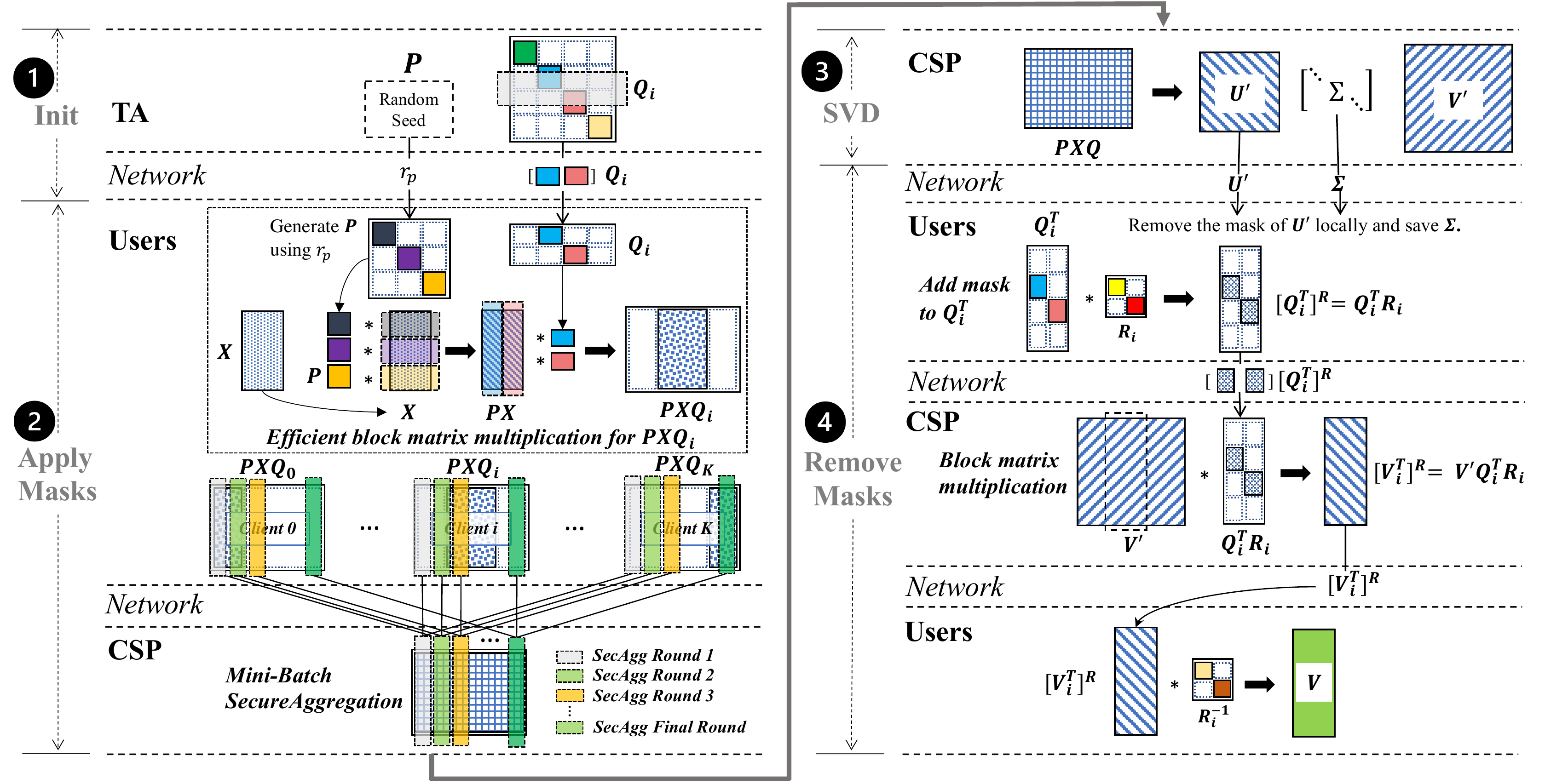}
	\caption{Detailed Workflow of \solution, which has four steps: Step \ding{202}: Trusted Authority (TA) initialize and send the masks to users. Step \ding{203}: Users apply masks and deliver the masked data to the computation service provider (CSP) through secure aggregation. Step \ding{204}: CSP conducts standard SVD on masked data. Step \ding{205}: Users remove the masks and get final results.}
	\label{fig:framework}
\end{figure*}

To solve this problem, we ask: \textit{Can we find a type of removable noise to protect data privacy as well as keep the data size unchanged to simultaneously achieve lossless accuracy and high efficiency?} Our answer is \solution. Briefly, 1) we propose a removable random mask delicately designed for SVD to protect privacy, and the masks could be completely removed from SVD results; 2) The masked data has the same size as the raw data, which results in no efficiency overhead during matrix decomposition. Meanwhile, we propose optimizations from algorithm and system aspects, including block-based mask generation, efficient data masking and data recovering, mini-batch secure aggregation, and advanced disk offloading, to further improve efficiency. With our delicate design, \solution could achieve lossless accuracy and high practicality over billion-scale data. Furthermore, we provide privacy analysis on \solution and show that \solution is highly confidential. In this section, we present the technical details of \solution.

\parab{Roles:} According to the different functionalities, we specify three types of roles in our system: 
 
\begin{ecompact}
	\item \textbf{Trusted Authority (TA)}: TA is responsible for generating removable secret masks and delivering them to the users. TA can remain offline once the system initialization is done.
	\item \textbf{Computation Service Provider (CSP)}: The CSP is responsible for running a standard SVD algorithm on the masked data and delivering the masked SVD results to the users.
	\item \textbf{Users}: The parties that own raw data (\ie, $\mathbf{X}$) and wish to run an SVD-based algorithm jointly.
\end{ecompact}

\parab{Workflow Overview:} \solution has the following four steps, which is also illustrated in \Cref{fig:framework} \label{sec:fedsvd_steps}:

\textbf{Step \ding{202}} : TA generates two removable random orthogonal masks $\mathbf{P} \in \mathbb{R}^{m \times m}$ and $\mathbf{Q} \in \mathbb{R}^{n \times n}$. Mask $\mathbf{P}$ is broadcasted to all users. The matrix $\mathbf{Q}$ is horizontally split into $k$ parts $\mathbf{Q}^T=[\mathbf{Q}_1^T,.,\mathbf{Q}_i^T,.,\mathbf{Q}_k^T]$, and TA sends $\mathbf{Q}_i$ to user-i. The detail of the removable random mask is introduced in \S\ref{sec:init}. To support billion-scale applications, we have proposed efficiency mask generation (\S\ref{sec:init}) and delivery method (\S\ref{sec:apply_mask}), reducing the computation and communication complexity from $O(n^3)$ and $O(m^2+n^2)$ to $O(n)$.
	
\textbf{Step \ding{203}} : All users compute $\mathbf{X}'_i=\mathbf{P}\mathbf{X}_i\mathbf{Q}_i$, where $\mathbf{X}'_i$ is the local masked data. The CSP gets $\mathbf{X'}$ through secure aggregation on $\mathbf{X}'_i$. To support billion-scale data, we propose efficient block matrix multiplication to reduce computation complexity from $O(m^2n+mn^2)$ to $O(mn)$ and mini-batch secure aggregation to reduce the memory usage at the server. More details are introduced in \S\ref{sec:apply_mask}.
	
\textbf{Step \ding{204}} :  CSP runs a standard SVD algorithm, factorizing $\mathbf{X}'$ into $\mathbf{U}' \mathbf{\Sigma} \mathbf{V}'^T$. We do not specify the algorithm (\eg, householder transformation) of solving the SVD problem, and \solution can work with any lossless SVD solver.

\textbf{Step \ding{205}} : Users downloads $\mathbf{U}', \mathbf{\Sigma}$, and recover $\mathbf{U}$ by $\mathbf{P}^T\mathbf{U}'$. $\mathbf{V}_i^T$ is jointly recovered under the protection of random masks between the CSP and users. We propose efficient mask removing of $\mathbf{V}_i^T$ via block matrix computation which reduces the complexity from $O(n_i^3)$ to $O(n_i)$. The details are introduced in \S\ref{sec:remove_masks}.

Organization of this section: \S\ref{sec:init} introduces the removable random mask delicately designed for SVD. \S\ref{sec:apply_mask} introduces the detail of mask initialization and applying the mask, \S\ref{sec:remove_masks} introduces the detail of removing mask. In \S\ref{sec:Offloading Memory Burden to Disks}, we propose an advanced disk offloading strategy according to the data access patterns. \S\ref{sec:privacy_proof} gives privacy analysis of \solution. 

\subsection{Removable Random Masks for SVD} \label{sec:init}

We propose a masking method that allows running SVD directly on the masked data and the masks could be removed from the results. Denoting the data matrix as $\mathbf{X}$, we use two random orthogonal matrices $\mathbf{P}$ and $\mathbf{Q}$ to mask the data $\mathbf{X}$ as $\mathbf{X}'=\mathbf{PXQ}$. \Cref{theorem:1} proves that $\mathbf{X}'$ has the same singular values with $\mathbf{X}$, the singular vectors of $\mathbf{X}$ and $\mathbf{X}'$ can be orthogonal transformed to each other using matrices $\mathbf{P}$ and $\mathbf{Q}$. Thus we can get the singular vectors of $\mathbf{X}$ by removing the masks from the singular vectors of $\mathbf{X}'$ (\textit{i.e.,} using orthogonal transformation).

\begin{theorem} \label{theorem:1}
	For an arbitrary matrix $\mathbf{X} \in \mathbb{R}^{m \times n}$ with SVD result $\mathbf{X}=\mathbf{U \Sigma V^T}$, we can use two random orthogonal matrices $\mathbf{P} \in \mathbb{R}^{m \times m}$ and $\mathbf{Q} \in \mathbb{R}^{n \times n}$ to mask $\mathbf{X}$ into $\mathbf{X}'=\mathbf{PXQ}$. Assuming the SVD result of $\mathbf{X
	'}$ is $\mathbf{U}'\mathbf{\Sigma}'\mathbf{V'}^T$. Then we can get SVD results of $\mathbf{X}$ through: $\mathbf{\Sigma}=\mathbf{\Sigma}'$, $\mathbf{U}=\mathbf{P}^T\mathbf{U}'$ and $\mathbf{V^T}=\mathbf{V'}^T\mathbf{Q}^T$.
\end{theorem}

\vspace{-3mm}

\begin{proof}
	By plugging $\mathbf{X}=\mathbf{U \Sigma V^T}$ into $\mathbf{X}'$, $\mathbf{X}'$ could be represented as $\mathbf{X}'=\mathbf{PXQ}=\mathbf{(PU) \Sigma} \mathbf{(V}^T\mathbf{Q)}$. According to \Cref{eq:pq_orth}, $\mathbf{PU}$ and $\mathbf{V}^T\mathbf{Q}$ are orthogonal matrices:
	\begin{equation}
		\begin{aligned} \label{eq:pq_orth}
			\mathbf{(PU)^{-1}=U^{-1}P^{-1}} &=\mathbf{U^TP^T=(PU)^T} \\
			\mathbf{(V^TQ)^{-1}=Q^{-1}(V^{T})^{-1}} &= \mathbf{Q^TV=(V^TQ)^T} \\
		\end{aligned}
	\end{equation}
	Then $\mathbf{(PU) \Sigma} \mathbf{(V}^T\mathbf{Q)}$ is the SVD result of $\mathbf{X}'$. Accordingly, $\mathbf{PU}=\mathbf{U}'$, $\mathbf{\Sigma}=\mathbf{\Sigma}'$, and $\mathbf{V}^T\mathbf{Q}=\mathbf{V'}^T$. Then $\mathbf{U}=\mathbf{P}^T\mathbf{U}'$ and $\mathbf{V^T}=\mathbf{V'}^T\mathbf{Q}^T$.
\end{proof}

\vspace{-2mm}

We present a random orthogonal matrix generation method in \Cref{alg:random_orthogonal} using the Gram-Schmidt process \cite{daniel1976reorthogonalization}. It is proved in prior work \cite{gupta2018matrix} that Gram-Schmidt process on Gaussian matrices produces uniformly distributed random orthogonal matrices.

\parab{Block-based Efficient Mask Generation:} However, the complexity of Gram-Schmidt process on a $n$ dimensional square matrix is $O(n^3)$ \cite{daniel1976reorthogonalization}. Thus we propose an efficient random orthogonal matrix generation algorithm through building blocks, which is presented in Algorithm \ref{alg:random_orthogonal_block}. Briefly, we decompose the problem of generating a $n$ dimensional orthogonal matrix into generating small orthogonal matrices with size $b$, placing these small matrices at the diagonal position, and forming a $n$ dimensional matrix. Then the complexity of generating $n$ dimensional orthogonal matrix reduces to $O(b^3\frac{n}{b})=O(b^2n)=O(n)$, where $b \ll n$.

\vspace{-2mm}

\begin{algorithm}[h]
    \small
	\caption{Generate Random Orthogonal Matrix}
	\label{alg:random_orthogonal}
	\KwIn{Dimension of the matrix $n$}
	\KwOut{Orthogonal matrix $\mathbf{Q} \in \mathbb{R}^{n \times n}$}
	\SetKwFunction{FMain}{Orthogonal}
	\SetKwProg{Fn}{Function}{:}{}
	\Fn{\FMain{$n$}}{
		Randomly sample matrix $\mathbf{R} \in \mathbb{R}^{n \times n}$, where $R_{i,j} \sim \mathcal{N}(0, 1)$\\
		$[\mathbf{Q}, \sim] = GramSchmidt(\mathbf{R})$ \\
		\Return $\mathbf{Q}$
	}
	\textbf{End Function}
	\vspace{-1mm}
\end{algorithm}

\vspace{-8mm}

\begin{algorithm}[!h]
    \small
	\caption{Efficient Orthogonal Matrix Generation Through Building Blocks}
	\label{alg:random_orthogonal_block}
	\KwIn{Dimension of the matrix $n$, size of building blocks $b$}
	\KwOut{Orthogonal matrix $\mathbf{Q} \in \mathbb{R}^{n \times n}$}
	\SetKwFunction{FMain}{EfficientOrthogonal}
	\SetKwProg{Fn}{Function}{:}{}
	\Fn{\FMain{$n, b$}}{
	    $\mathbf{Q}$ $\leftarrow$ [], $i\leftarrow0$ \\
	    \While{$i<n$}{
	        $b' \leftarrow min(b, n-i)$ ~\\
	        $\mathbf{Q}_b \leftarrow Orthogonal(b')$ \tcp{\Cref{alg:random_orthogonal}}
	        $\mathbf{Q} \leftarrow \begin{bmatrix}\mathbf{Q} & \mathbf{0} \\ \mathbf{0}  & \mathbf{Q_b} \\ \end{bmatrix}$, $i\leftarrow i+b'$ \\
	    }
	    \Return $\mathbf{Q}$
	 }
	\textbf{End Function}
	\vspace{-1mm}
\end{algorithm}

\vspace{-4mm}

\parab{Block size controls the trade-off between efficiency and privacy protection} It is worth noting that the block size (\ie, $b$) simultaneously impacts the system privacy protection and efficiency. Theoretically, large block size increases the freedom of the masks, thus increases the effectiveness of privacy protection. Meanwhile, large block size increases the computation overhead, thus decreases the system efficiency. We have reported attack experiments using the SOTA attack method in \S\ref{sec:attack} showing that the attack fails in recovering valid information as long as $b$ is large enough. We set $b=1000$ in our experiments since our attacking experiments on many datasets show that $1000$ is a good choice of gaining enough privacy protection and benefiting from the efficiency brought by the block-based optimizations. The proper block size may differ on different datasets, and we suggest adjusting block size according to the datasets in the application, which is also discussed in \S\ref{sec:attack}.

\subsection{Initialization \& Applying the Masks} \label{sec:apply_mask}

We propose a federated computation process based on the removable random masks to apply masks on the raw data. At the beginning of the computation, TA holds masks $\mathbf{P,Q}$ and users jointly hold $\mathbf{X}=[\mathbf{X}_1, \mathbf{X}_2,...,\mathbf{X}_k]$. At the end of the computation, CSP receives $\mathbf{X}'=\mathbf{PXQ}$ and does not learn any other information.

Equation \ref{eq:x'} shows our idea of federally computing $\mathbf{X}'$. According to the rule of block matrix multiplication, we can decompose $\mathbf{PXQ}$ into $\sum_{i=1}^k \mathbf{P}\mathbf{X}_i\mathbf{Q}_i$.

\vspace{-5mm}
\begin{equation} \label{eq:x'}
		\mathbf{X'} = \mathbf{PXQ =P [X_1,.,X_i,.,X_K] [Q_1^T,.,Q_i^T,.,Q_k^T]^T} = \sum_{i=1}^{k} \mathbf{PX_iQ_i} \\
\end{equation}
\vspace{-3mm}

Thus the federated computation of $\mathbf{X}'$ can be divided into two steps. \textbf{Step} \ding{202}: TA broadcasts $\mathbf{P}$ to all users, then horizontally splits the mask $\mathbf{Q}$ into $\{\mathbf{Q}_i \in \mathbb{R}^{n_i \times n}|1\le i \le k\}$, and sends $\mathbf{Q}_i$ to user-$i$. \textbf{Step} \ding{203}: Users compute $\mathbf{PX_iQ_i}$, the CSP runs a secure aggregation to get $\sum_i \mathbf{P} \mathbf{X}_i \mathbf{Q}_i$. The secure aggregation conceals the intermediate results (\textit{i.e.,} $\mathbf{P}\mathbf{X}_i\mathbf{Q}_i$), and guarantees that CSP only learns $\mathbf{X}'$.

\parab{Communication Efficient Mask Delivery:} We observe that directly transferring $\mathbf{P,Q}$ has $O(m^2+n^2)$ communication complexity. Based on \Cref{alg:random_orthogonal_block}, we propose to reduce the communication complexity through transferring only one random number or small blocks of the mask. More specifically, the TA only broadcast a random seed $r_p$ for mask $\mathbf{P}$ since Gram-Schmidt is a deterministic algorithm that yields the same orthogonal matrix as long as the input matrices are the same, thus the users can generate $\mathbf{P}$ locally using the same random seed. TA only sends the sliced matrix blocks for mask $\mathbf{Q}$ and the zeros are omitted during the transmission. In summary, communication complexity of transferring $\mathbf{P,Q}$ are reduced to $O(1)$ and $O(b^2 \frac{n}{b})=O(n)$.

\parab{Efficient Data Masking via Block Matrix Multiplication:} We observe that our data masking process (\ie, computing $\mathbf{PXQ}$) has cubic complexity (\ie, $O(m^2n + mn^2)$) which brings large computation overhead in large-scale applications. To reduce the complexity, we adopt block matrix multiplications since $\mathbf{P,Q}$ are sparse matrices and consist of blocks. A concrete example is presented in \Cref{eq:example data masking}, where the zeros are omitted in the computation. After adopting the block matrix multiplication, the data masking complexity is reduced from cubic complexity to $O(\frac{m}{b}*b^2*n + \frac{n}{b}*b^2*m)=O(mn)$.

\vspace{-2mm}
\begin{equation} \label{eq:example data masking}
	\begin{bmatrix}
		\mathbf{P_1} & 0 & 0 \\
		0 & \mathbf{P_2}  & 0 \\ 
		0 & 0 & \mathbf{P_3} \\
	\end{bmatrix}
	\begin{bmatrix}
		\mathbf{X_1} \\
		\mathbf{X_2} \\ 
		\mathbf{X_3} \\
	\end{bmatrix}
	=
	\begin{bmatrix}
		\mathbf{P_1X_1} \\
		\mathbf{P_2X_2} \\ 
		\mathbf{P_3X_3} \\
	\end{bmatrix}
\end{equation}
\vspace{-2mm}

\parab{Memory Efficient Mini-batch Secure Aggregation:} We observe that secure aggregation (SecAgg) directly processes the whole data matrix (\ie, $\mathbf{X}'_i=\mathbf{P}\mathbf{X}_i\mathbf{Q}_i$), and it will bring significant memory burden to the server and users in \solution since $\mathbf{X}'_i$ is a large matrix. We propose to split $\mathbf{X}'_i$ into batches and only process one batch of data in each round of SecAgg. Mini-batch SecAgg works because the aggregations of different rows or columns of $\mathbf{X}'_i$ are independent.

\subsection{Removing the Masks} \label{sec:remove_masks}

Intuitively, the masks in the final results could be removed by each user locally if the CSP broadcast $\mathbf{U}'$ and $\mathbf{V}'^T$, \ie $\mathbf{U}=\mathbf{P}^T\mathbf{U}'$ and $\mathbf{V'}_i^T=\mathbf{Q}_i^T\mathbf{V'}^T$. However, $\mathbf{V}'^T$ contains masked eigenvectors of all users, sending $\mathbf{V}'^T$ from CSP to users may bring privacy issues because users hold more information than CSP, \eg, $\mathbf{Q}_i$. Thus we propose a federated computation process to recover $\mathbf{V}'^T$. For $\mathbf{U}'$, users can remove the mask locally because $\mathbf{U}$ is the defined as the shared result in federated SVD (\ie, \S\ref{sec:federated_svd}).

During the recovery of $\mathbf{V}'^T$, we want to guarantee the confidentiality of both $\mathbf{Q}_i^T$ and $\mathbf{V}'$, \textit{i.e.,} the users cannot get the whole $\mathbf{V}'$ matrix and the CSP cannot learn $\mathbf{Q}_i^T$.

Our solution is first masking $\mathbf{Q}_i^T$ using another random matrix $\mathbf{R}_i \in \mathbb{R}^{n_i \times n_i}$ according to \Cref{eq:recover_v}. Then user $i$ sends the $[\mathbf{Q_i^T}]^R$ (\textit{i.e.,} the masked $\mathbf{Q}_i^T$) to the CSP, which will subsequently compute $[\mathbf{V}_i^T]^R$ and send $[\mathbf{V}_i^T]^R$ back to user $i$. Then user $i$ can remove the random mask according to \Cref{eq:recover_v} and get the final result (\textit{i.e.,} $\mathbf{V}_i^T$).

\vspace{-3mm}
\begin{equation} \label{eq:recover_v}
	[\mathbf{Q_i^T}]^R = \mathbf{Q_i^T R_i},[\mathbf{V_i^T}]^R = \mathbf{V'^T} [\mathbf{Q_i^T}]^R,\mathbf{V_i^T} = [\mathbf{V_i^T}]^R \mathbf{R_i^{-1}}
\end{equation}
\vspace{-3mm}

It is worth noting that $\mathbf{V}^T_i$ also could be recovered through $\mathbf{V}^T_i=\mathbf{\Sigma}^{-1}_n\mathbf{U}_n^T\mathbf{X}_i$, where $\mathbf{\Sigma}^{-1}_n$ and $\mathbf{U}_n^T$ mean the first $n$ rows of $\mathbf{\Sigma}^{-1}$ and $\mathbf{U}^T$. However, this method only works when $m>=n$. When $m<n$, we can only recover the first $m$ rows of $\mathbf{V}^T_i$ but not the full matrix. Thus this method is not a general solution.

\parab{Efficient Recovery of $\mathbf{V}^T$ via Block Matrix Computation:} We observe that although $\mathbf{Q}_i^T$ is a sparse matrix consisting of blocks according to \Cref{alg:random_orthogonal_block}. However, the computing and transferring $\mathbf{Q}_i^T \mathbf{R}_i$ is costly since $\mathbf{R}_i$ is a dense random matrix. The computation and communication of $\mathbf{Q}_i^T \mathbf{R}_i$ has $O(\frac{n_i}{b}b^2n_i)=O(n_i^2)$ complexity. To improve efficiency, our solution generates $\mathbf{R}_i$ through putting a bunch of square random matrix diagonally and the size of each small random matrix is decided by $\mathbf{Q}_i^T$, such that $\mathbf{Q}_i^T \mathbf{R}_i$ is still a sparse matrix consists of blocks, \Cref{eq:example1} shows an example. The complexity is reduced from $O(n_i^2)$ to $O(\frac{n_i}{b}b^3)=O(n_i)$. Moreover, since $\mathbf{R}_i$ is consist of block matrices, the complexity of computing its inverse (\ie, $\mathbf{R_i^{-1}}$) also reduces from $O(n_i^3)$ to $O(n_i)$.

\vspace{-3mm}
\begin{equation} \label{eq:example1}
	\mathbf{Q_i^T R_i}=\begin{bmatrix}
		0 & 0 \\
		\mathbf{Q_{i,1}^T} & 0 \\
		0 & \mathbf{Q_{i,2}^T} \\ 
		0 & 0 & \\
	\end{bmatrix}
	\begin{bmatrix}
		\mathbf{R_i^1} & 0 \\
		0 & \mathbf{R_i^2} \\
	\end{bmatrix}
	=
	\begin{bmatrix}
		0 & 0 \\
		\mathbf{Q_{i,1}^T}\mathbf{R_i^1} & 0 \\
		0 & \mathbf{Q_{i,2}^T}\mathbf{R_i^2} \\ 
		0 & 0 & \\
	\end{bmatrix}
\end{equation}

\subsection{Disk Offloading via Data Access Patterns}
\label{sec:Offloading Memory Burden to Disks}

\parab{Observation:} Dealing with large-scale matrices usually requires large hardware memory. For example, a 100K $\times$ 1M 64-bit matrix requires approximately 745GB RAM, making the memory space for computation very limited. A standard solution is offloading part of the memory storage to disks using swap memory and reloading the data when needed. The operating system (OS) will automatically schedule the disk offloading. However, naively following the OS scheduling with no specific design for our algorithm is inefficient.

\parab{Solution:} We propose an advanced disk offloading strategy according to the data access pattern for \solution. 1) Offloading strategy for $\mathbf{P,Q}$. According to our observation, $\mathbf{P,Q}$ are used twice in the computation when applying and removing the masks. Hence, on the client-side, we immediately save the blocks of $\mathbf{P,Q}$ to disk when they are generated or received from TA. When applying or removing the masks, we load and use $\mathbf{P,Q}$ block by block (\ie, sequentially). And each block will be removed from the memory when its computation finishes; 2) Offloading strategy for large dense data matrices (\eg, $\mathbf{X},\mathbf{PXQ},\mathbf{U},\mathbf{V}^T$). We store the large data matrices in disk and leave a file map in memory. The file map will automatically read the needed matrix components. However, direct adoption of file maps may bring severe efficiency issues. The file map uses consistent storage on the disk and the matrix is stored by rows by default. If the manner we access the matrix conflicts with the storage manner (\eg, access by column), the efficiency will be very low. Thus we have optimized the implementation such that all the file-map matrices are stored adaptively according to the access pattern. The evaluation results show that our advanced disk offloading strategy reduces the time consumption by 44.7\% compared with using swap memory scheduled by OS. Detail could be found in \S\ref{sec:optimization}.

\subsection{Privacy Analysis} \label{sec:privacy_proof}

In this section, we analyze the confidentiality of \solution. We consider the TA to be a fully trusted entity, while the CSP and users are semi-honest parties. This means, both the CSP and users will honestly follow the pre-designed protocols but also attempt to infer private data. We also assume there is no collusion between the CSP and the users.

\parab{CSP cannot reveal the original matrix:} According to \Cref{fig:framework}, the total messages received by the CSP are $\mathbf{X}'=\sum_i [\mathbf{X}'_i]^R$ and $[\mathbf{Q}_i^T]^R$. 1) According to the prior work \cite{bonawitz2016practical}, the CSP only learns the aggregated results $\mathbf{X}'$, no information is leaked during the secure aggregation; 2) In \Cref{theorem:privacy_proof}, we show that there is an infinite number of raw data that could be masked into the same matrix. If the CSP has no prior knowledge about the data distribution, it can never recover the true data because the true data is not identifiable. Alternatively, the CSP can empirically choose data distribution as prior knowledge and perform attacks \cite{ica_attack} on the masked data. However, the attack experiments in \S\ref{sec:attack} show that the attack fails in getting valid information if we set the hyper-parameter properly; 3) According to the prior work \cite{zhang2019secure}, the masked data $[\mathbf{Q}_i^T]^R$ cannot be computationally distinguished from a random matrix, thus leaks no information.

In conclusion, \solution is secure against CSP which cannot reveal the raw data.

\begin{theorem} \label{theorem:privacy_proof}
	Given a masked data $\mathbf{X}'=\mathbf{P}_1 \mathbf{X}_1 \mathbf{Q}_1$, there are infinite number of raw data $\mathbf{X}_2$ that can be masked into $\mathbf{X}'$, \ie, $\mathbf{P}_2 \mathbf{X}_2 \mathbf{Q}_2 =\mathbf{P}_1 \mathbf{X}_1 \mathbf{Q}_1=\mathbf{X}'$. 
\end{theorem}

\begin{proof} Given two random orthogonal matrix $\mathbf{R}_1 \in \mathbb{R}^{m \times m}$ and $\mathbf{R}_2 \in \mathbb{R}^{n \times n}$, we can rewrite $\mathbf{X}'$ into 
\begin{equation}
\begin{aligned}
	\mathbf{X'} & = \mathbf{P_1X_1Q_1=P_1U \Sigma V^TQ_1=P_1U(R_1^TR_1) \Sigma (R_2R_2^T)V^TQ_1} \\
	&=\mathbf{(P_1UR_1^T)(R_1\Sigma R_2)(R_2^TV^TQ_1)} \\
\end{aligned}
\end{equation}

Let $\mathbf{X}_2=\mathbf{R}_1\mathbf{\Sigma}\mathbf{R}_2$, $\mathbf{P}_2=\mathbf{P}_1\mathbf{U}\mathbf{R}_1^T$, $\mathbf{Q}_2=\mathbf{R}_2^T\mathbf{V}^T\mathbf{Q}_2$, then we get $\mathbf{P}_2\mathbf{X}_2\mathbf{Q}_2=\mathbf{P}_1\mathbf{X}_1\mathbf{Q}_1=\mathbf{X}'$. $\mathbf{R}_1,\mathbf{R}_2$ are random orthogonal matrices and the number of orthogonal matrices with certain size is infinite in real number field, thus we have infinite number of $\mathbf{X}_2$ that also can be masked into $\mathbf{X}'$ and the CSP cannot identify the real data.
\end{proof}

\vspace{-2mm}

\parab{The users can only learn the final results:} According to \Cref{fig:framework}, the user i receives: $\mathbf{P}, \mathbf{Q}_i, \mathbf{U}', \mathbf{\Sigma}, [\mathbf{V}_i^T]^R$, and the valid information are $\mathbf{U}, \mathbf{\Sigma}, \mathbf{V}_i$, which are exactly the final results of the federated SVD problem defined in \S\ref{sec:federated_svd}. Thus each user only learn its final results and receives nothing about other users' private data. Additionally, \solution is secure against collusion between the users because a group of cooperated users could be treated as a single user who owns more local data, and they cannot obtain the privacy of other users outside the group.

\parab{The TA learns nothing:} Since TA receives nothing in \Cref{fig:framework} and remains offline after initialization, it learns nothing in the algorithm.

In summary, \solution is secure against CSP, TA receives nothing during the computation, and the users only get their final results. \solution is highly confidential.

\section{Applications Based on FedSVD}

Based on \solution, we propose three applications: principal component analysis (PCA), linear regression (LR), and latent semantic analysis (LSA). All these applications have the same first three steps with \solution and only differ at the last step. Tailored optimizations are also made for each application to further improve efficiency.

\vspace{-3mm}
\begin{figure}[h!]
	\centering
	\includegraphics[scale=0.4]{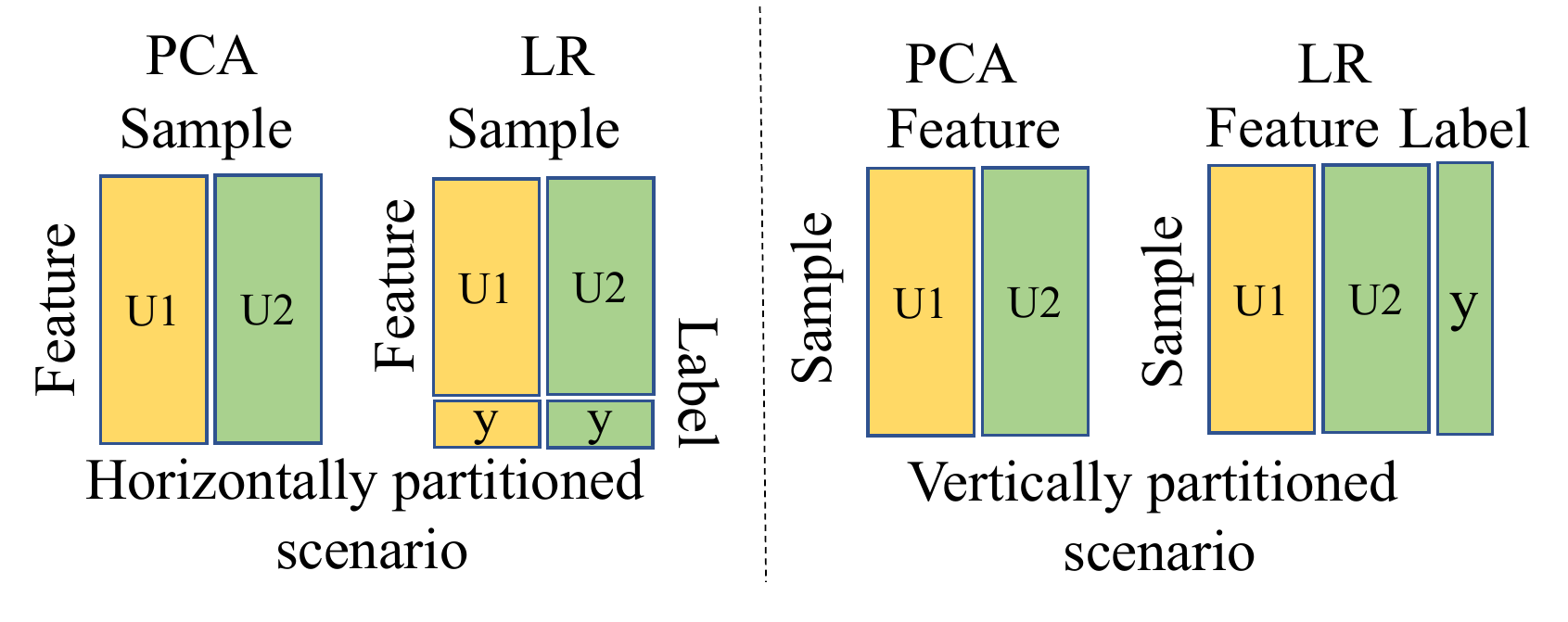}
	\caption{Federated PCA and LR under different data settings.}
	\label{fig:application_pca_lr}
\end{figure}
\vspace{0.5mm}

\parab{PCA in horizontally partitioned scenario:} PCA in federated learning setting typically has two data partition schemas,~\ie, horizontally and vertically, which are illustrated in \Cref{fig:application_pca_lr}. In this paper, we consider the horizontal federated PCA since it is the most common data setting in \highlight{medical and biometric} studies in which multiply institutions have the same feature on different samples. Given a normalized matrix $\mathbf{X}$, PCA decomposes it into $\mathbf{X} = \mathbf{U}_r {\Sigma}_r \mathbf{V}^T_r$, where $r$ is the number of principal components in PCA, $\mathbf{U}_r \in \mathbb{R}^{m \times r}$ and $\mathbf{V}^T_r \in \mathbb{R}^{r \times n}$ are the top-$r$ singular vectors with largest singular values. Such decomposition is also called truncated SVD. Considering PCA in horizontally partitioned scenario, the PCA result for user i is $\mathbf{U}_r^T\mathbf{X}_i \in \mathbb{R}^{r \times n_i}$. Accordingly, in \solution-based PCA, CSP only calculates and broadcasts the masked $\mathbf{U}'_r$ to all users and ignores the computation and transmission of $\mathbf{\Sigma},\mathbf{V}'^T$ to improve efficiency.

\parab{LR in vertically partitioned scenario:} LR in federated learning setting also has two data partition schemas,~\ie, horizontally and vertically, which are illustrated in \Cref{fig:application_pca_lr}. In this paper, we consider the vertical federated LR since it is the most common scenario of federated risk management and marking in the real-world applications \cite{yang2019federated}, in which different institutions hold different features on the same samples. Given a data matrix $\mathbf{X}=[\mathbf{X}_0;b] \in \mathbb{R}^{m \times n}$ and label $\mathbf{y}$, where $\mathbf{b}$ is the bias term, LR try to find a vector $\mathbf{w} \in \mathbb{R}^{n}$ such that $\mathbf{y=Xw}$. $\mathbf{w}$ could be solved through SVD on $\mathbf{X}$ and $\mathbf{w}=\mathbf{V}\mathbf{\Sigma}^{-1}\mathbf{U}^T\mathbf{y}$. In \solution-based LR, the user add mask to label $\mathbf{y}$ through $\mathbf{y}'=\mathbf{Py}$, then upload the masked label to CSP, which will subsequently compute $\mathbf{w}'=\mathbf{V}'\mathbf{\Sigma}^{-1}(\mathbf{U}')^T\mathbf{y}'=\mathbf{Q}^T \mathbf{V} \mathbf{\Sigma}^{-1} \mathbf{U}^T \mathbf{y}=\mathbf{Q}^T\mathbf{w}$. Then CSP broadcast the masked parameter matrix $\mathbf{w}'$ to all users, and each user can get the local parameters through $\mathbf{w}_i=\mathbf{Q}_i\mathbf{w}'$, where $\mathbf{w}_i \in \mathbb{R}^{n_i}$. In our LR design, the CSP will only broadcast the masked parameters and the $\mathbf{U}',\mathbf{\Sigma}$ and $\mathbf{V}'^T$ are not transmitted to improve the communication efficiency.

\parab{LSA:} Federated LSA is not sensitive to the data partition schemas since there is no clear definition of sample and feature in LSA. Briefly, LSA decomposes a data matrix $\mathbf{X} \in \mathbb{R}^{m \times n}$ (\eg, word-document matrix) into $\mathbf{X} = \mathbf{U}_r \mathbf{\Sigma}_r \mathbf{V}^T_r$, where $r$ is the number of embedding feature in LSA and $\mathbf{\Sigma}_r$ are the top-$r$ singular values. After the decomposition, $\mathbf{U}_r$ and $\mathbf{V}^T_r$ are treated as embedding features and used in the subsequent tasks, \eg, computing the similarity of different documents in NLP. Accordingly, in \solution-based LSA, the CSP and users run the same protocol of recovering $\mathbf{U}'$ and $\mathbf{V}'^T$ to recover their first $r$ vectors with the largest singular values, and the vectors outside $r$ are ignored to improve the efficiency.

\section{Experiments}

In this section, we provide a comprehensive evaluation of \solution regarding the lossless and efficiency on SVD task (\S\ref{sec:evaluation_svd}) and three applications (\S\ref{sec:evaluation_app}). Then we present the attack experiments in \S\ref{sec:attack}. Lastly, we show the effectiveness of the proposed system optimizations in \S\ref{sec:optimization}.

\subsection{Experiment Settings}

We have used five datasets in our experiments: MNIST \cite{lecun1998gradient}, Wine \cite{Dua:2019}, MovieLens-100K \cite{ml100k}, MovieLens-25M \cite{ml100k}, and synthetic data \cite{FedPCA}. We compare \solution with three state-of-the-art models: WDA-PCA \cite{WDA} which is a distributed rank-$k$ PCA method, FedPCA \cite{FedPCA} which is a federated $(\epsilon,\delta)$-differentially private PCA method, and PPD-SVD \cite{PPD-SVD} which is a HE-based distributed SVD method. In particular, on LR application, we compare \solution with two well-known federated LR solution: FATE \cite{liu2021fate} and SecureML \cite{Secureml}. We set $b=1000$ in \solution and $\epsilon=0.1,\delta=0.1$ for DP-based method. By default, following the prior work~\cite{han2009privacy,yang2019federated,Secureml}, we uniformly partition the data on two users, and partitioning data to more users will not impact our evaluations. Due to the space limitation, we put the detailed experiment setting in the \Cref{appendix:dataset}.

\subsection{Evaluation on SVD} \label{sec:evaluation_svd}

\parab{Lossless:} We have proved in Theorem \ref{theorem:1} that the masking-based protection in \solution is lossless. Here we would like to use more experimental results to show that the precision of \solution is lossless in the implementation.

We compare the precision of \solution with FedPCA on SVD tasks. The precision of SVD is measured by calculating the distance of singular vectors \cite{FedPCA} between the proposed methods and the standalone SVD. We use root-mean-square-error (RMSE) as the distance metric. \Cref{tab:lossless_evaluation} shows the results. \solution has about 10 orders of magnitude smaller error compared with DP-based solution.

To give a more straightforward understanding, we also evaluate the reconstruction error of \solution, \textit{i.e.,} distance to the raw data: $||\mathbf{X}-\mathbf{U} \mathbf{\Sigma} \mathbf{V}^T||$. Using mean absolute percentage error as the metric, \solution's reconstruction error is only 0.000001\% of the raw data. It is worth noting that \solution's tiny deviation in the experiment is brought by the floating number representation in computers. Theoretically, as proved in Theorem \ref{theorem:1}, \solution is lossless.

\vspace{-3mm}
\begin{figure}[h!]
	\small
	\centering
	\subfigure[Comparing to HE-based method on SVD task and billion scale data.]{
		\centering
		\label{fig:large_scale_svd_time_with_ppdsvd}
		\includegraphics[width=0.44\linewidth]{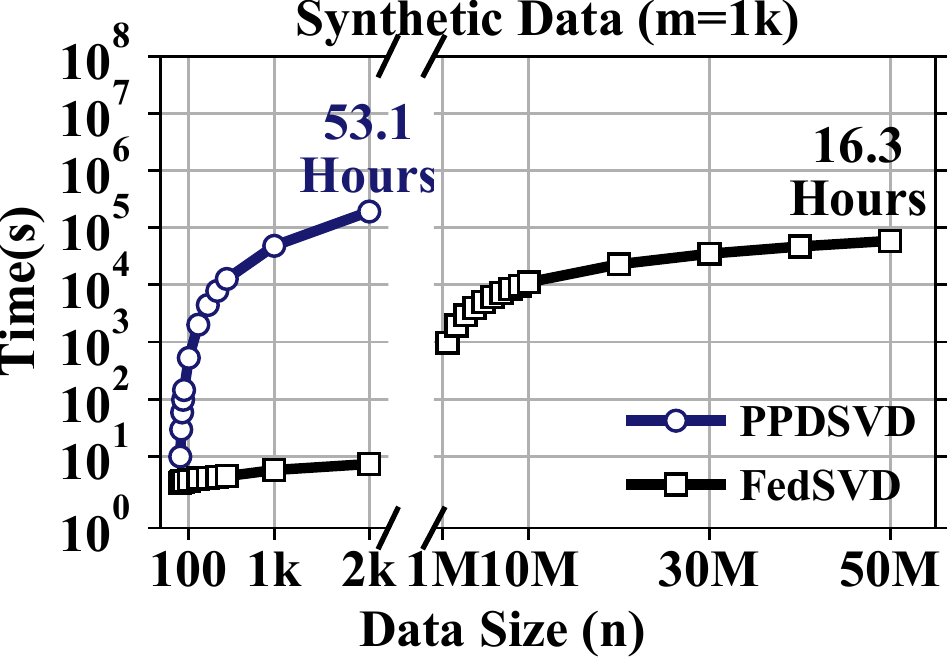}
	}
	\hfil
	\subfigure[Comparing to HE-based method on communication size.]{
		\centering
		\label{fig:small_scale_svd_comm}
		\includegraphics[width=0.44\linewidth]{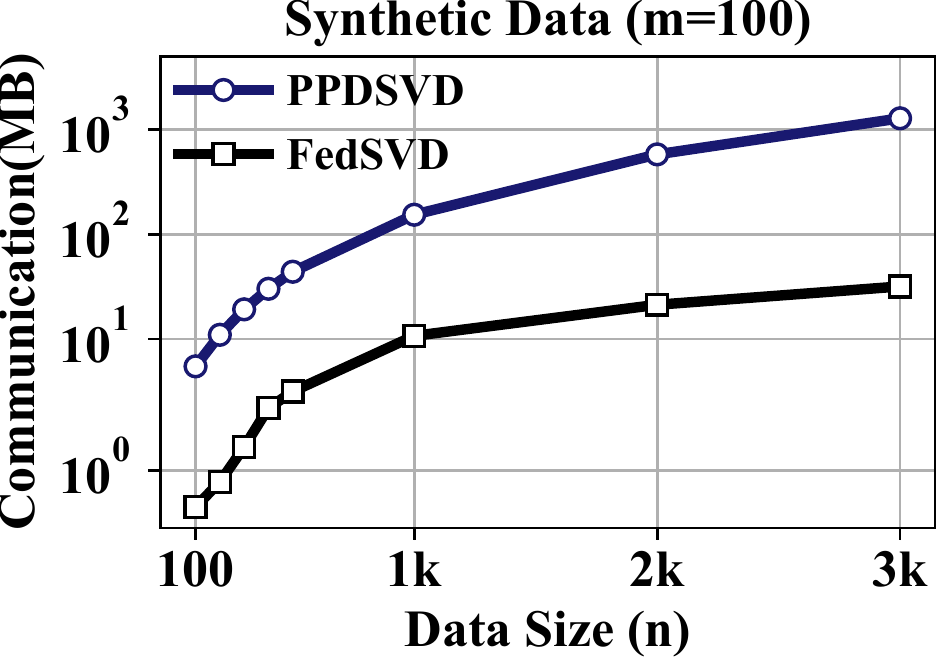}
	}
	\vspace{-2mm}
	\subfigure[Impact of network bandwidth on SVD efficiency.]{
		\centering
		\label{fig:time_to_bandwidth_svd}
		\includegraphics[width=0.46\linewidth]{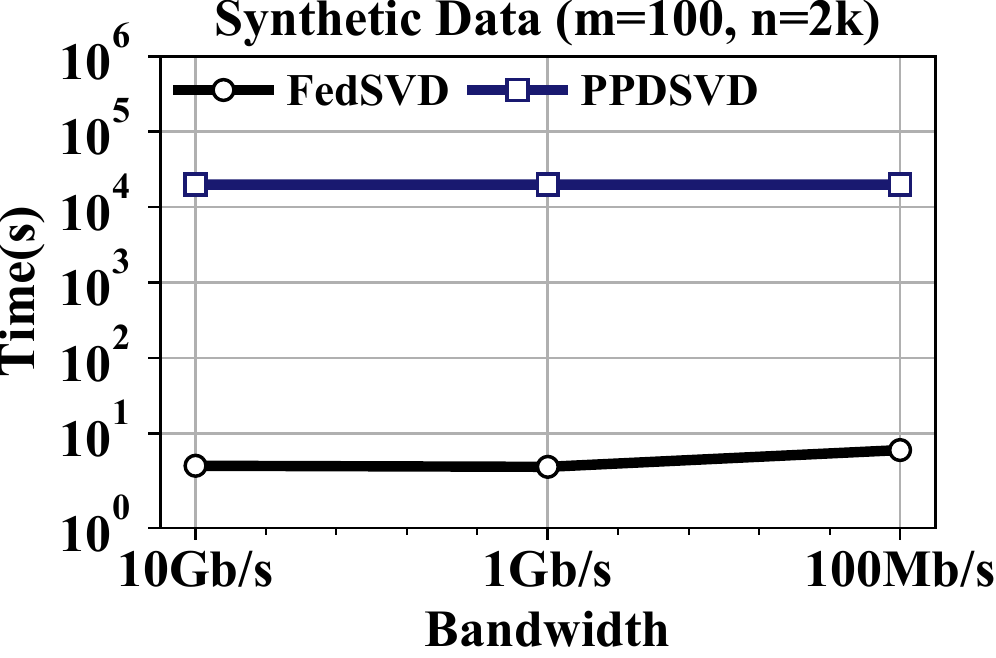}
	}
	\hfil
	\subfigure[Impact of network latency on SVD efficiency.]{
		\centering
		\label{fig:time_to_rtt_svd}
		\includegraphics[width=0.44\linewidth]{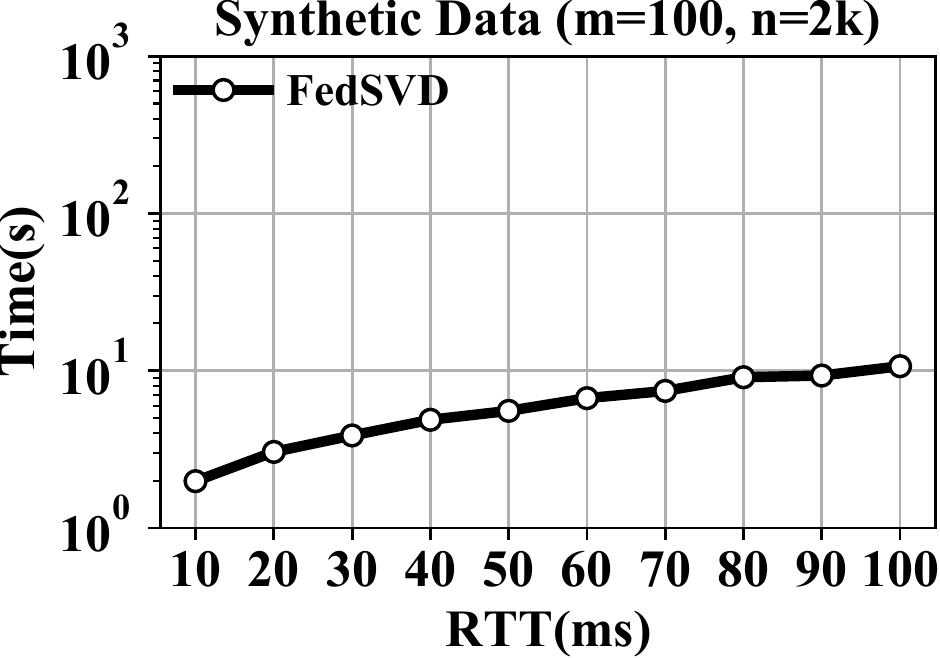}
	}
	\vspace{-2mm}
	\subfigure[Impact of block size on \solution's efficiency.]{
		\centering
		\label{fig:block_size}
		\includegraphics[width=0.44\linewidth]{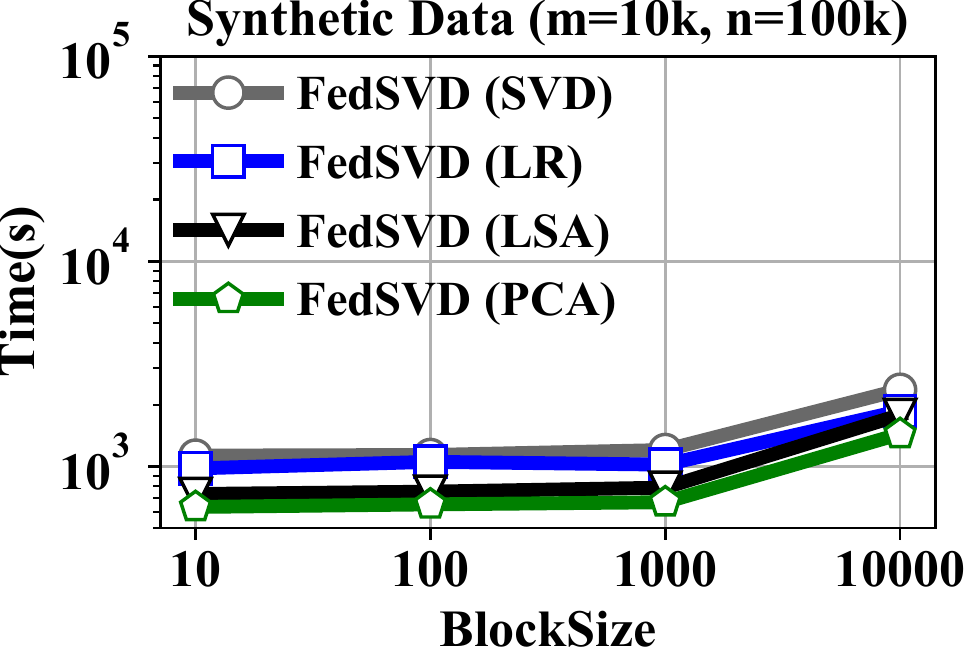}
	}
	\hfil
	\subfigure[Communication size of \solution under different \# of users and data size.]{
		\centering
		\label{fig:client_to_communication_svd}
		\includegraphics[width=0.43\linewidth]{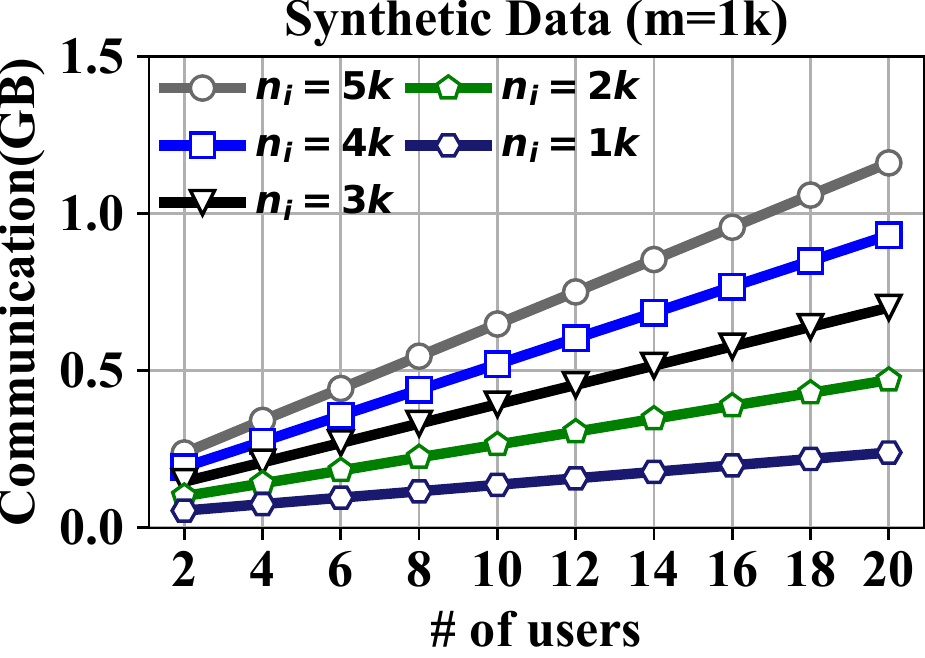}
	}
	\caption{Evaluation on SVD task.}
	\label{fig:svd_evaluation}
\end{figure}

\begin{table*}[h!]
	\centering
	\setlength{\tabcolsep}{0.4em}
	\renewcommand\arraystretch{0.85}
	\small
	\caption{Lossless evaluation on SVD task and three applications.}
	\begin{tabular}{|c|c|c|c|c|c|c|c|c|c|}
		\hline
		& \multicolumn{2}{c|}{SVD} & \multicolumn{3}{c|}{PCA / LSA Applications} & \multicolumn{4}{c|}{LR Application} \\
		\cline{2-10}
		\multirow{2}{*}{\tabincell{c}{Datasets}} & \tabincell{c}{FedPCA} & \solution & \tabincell{c}{FedPCA} & WDA & \solution & \tabincell{c}{SGD (10 Epoch) \\ (FATE \& SML)} & \tabincell{c}{SGD (100 Epoch) \\ (FATE \& SML)} & \tabincell{c}{SGD (1000 Epoch) \\ (FATE \& SML)} & \solution \\\hline
		Wine & $3.25*10^{-1}$ & $5.51*10^{-10}$ & 1.68 & 1.69 & $1.37*10^{-10}$ & 1.04 & 0.767 & 0.666 & 0.539 \\\hline
		MNIST & $9.37*10^{-2}$ &  $1.99*10^{-10}$ & $5.34*10^{-2}$ & $5.97*10^{-3}$ & $2.79*10^{-14}$ & 48.7 & 5.53 & 3.78 & 3.19 \\\hline
		ML100K & $7.95*10^{-2}$ & $1.45*10^{-13}$ & $4.45$ & $6.02*10^{-1}$ & $1.11*10^{-14}$ & 127 & 53.8 & 45.1 & 43.9 \\\hline
		Synthetic & $1.79*10^{-1}$ & $9.03*10^{-12}$ & $4.45$ & $9.13*10^{-4}$ & $9.09*10^{-15}$ & 1.71 & 0.974 & 0.849 & 0.813 \\\hline
	\end{tabular}
	\label{tab:lossless_evaluation}
	\vspace{-4mm}
\end{table*}

\begin{figure*}
  \begin{minipage}[h]{0.74\linewidth}
    \small
	\centering 
	\subfigure[Comparing \solution with FATE and SecureML on billion-scale data.]{
		\label{fig:large_scale_lr}
		\includegraphics[scale=0.4]{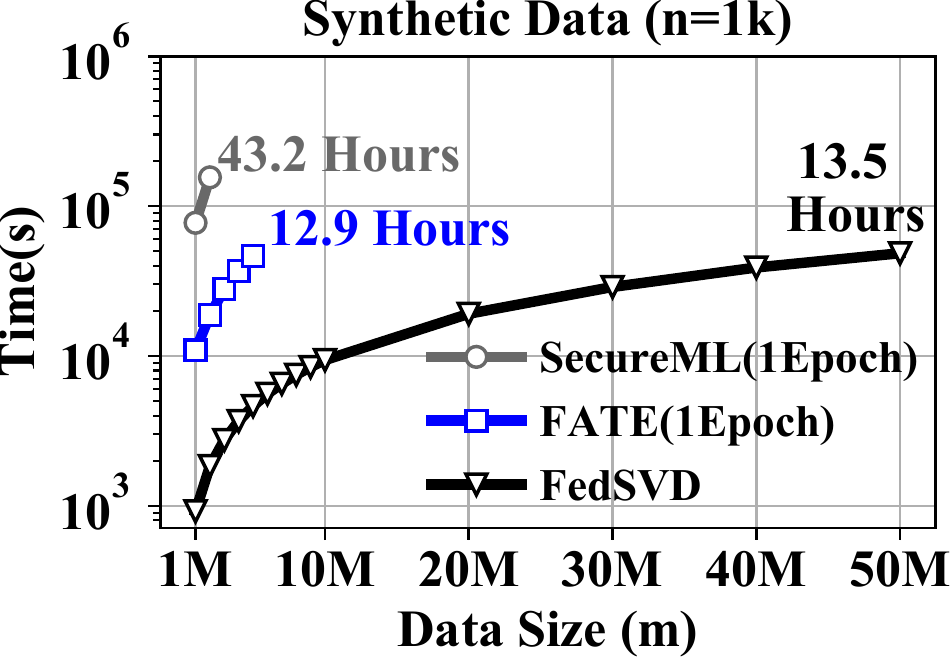}
	}
	\hfil
	\subfigure[Impact of network bandwidth on LR efficiency.]{
		\label{fig:time_to_bandwidth_lr}
		\includegraphics[scale=0.4]{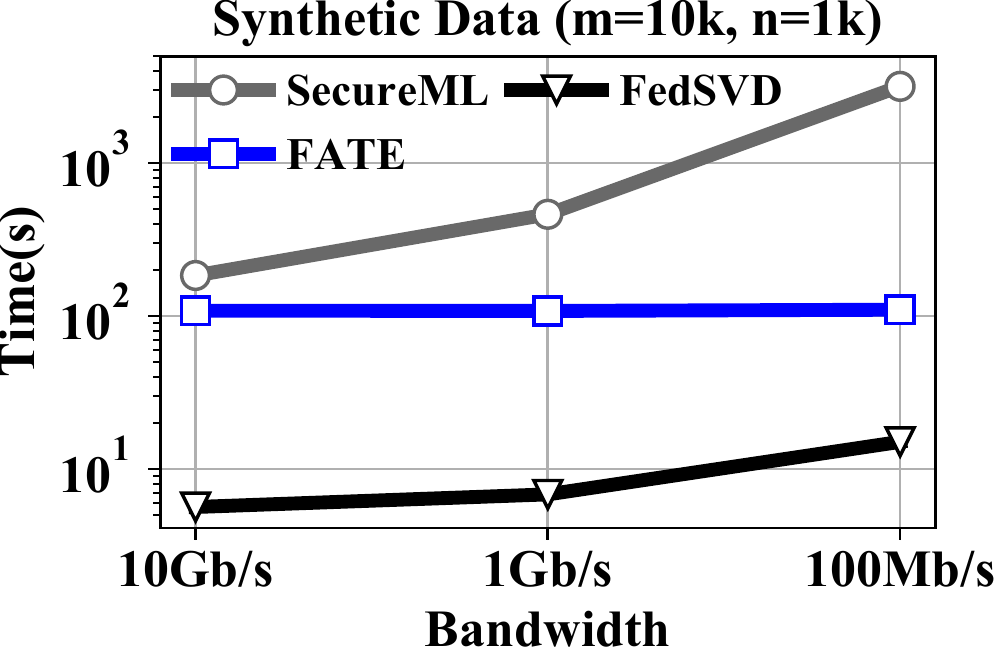}
	}
	\hfil
	\subfigure[Impact of network latency on LR efficiency.]{
		\label{fig:time_to_rtt_lr}
		\includegraphics[scale=0.4]{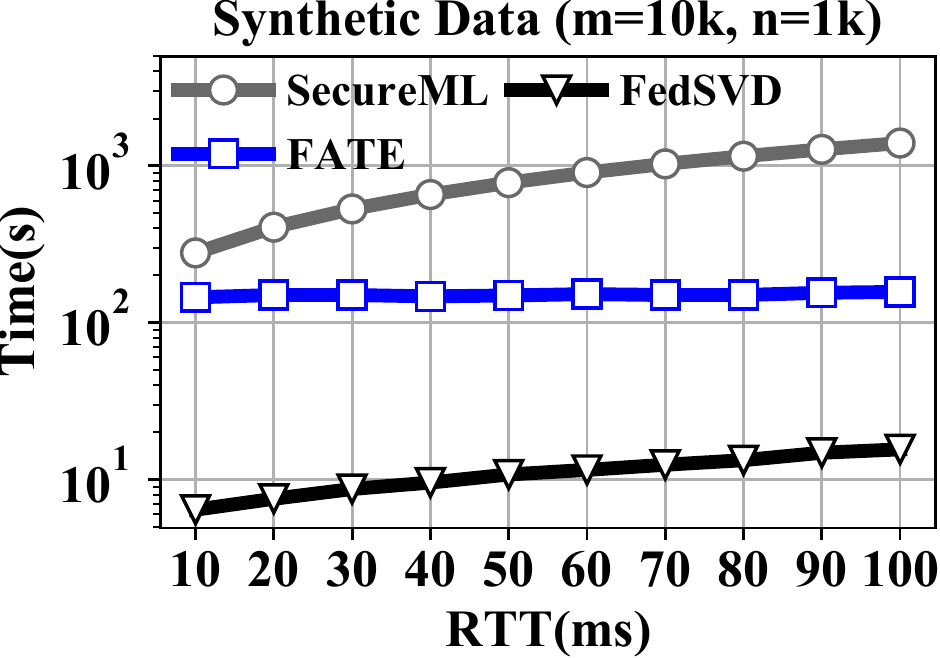}
	}
    \caption{Evaluation on LR Application.}
  \end{minipage}
  \begin{minipage}[h]{0.24\linewidth}
    \centering
    \includegraphics[scale=0.4]{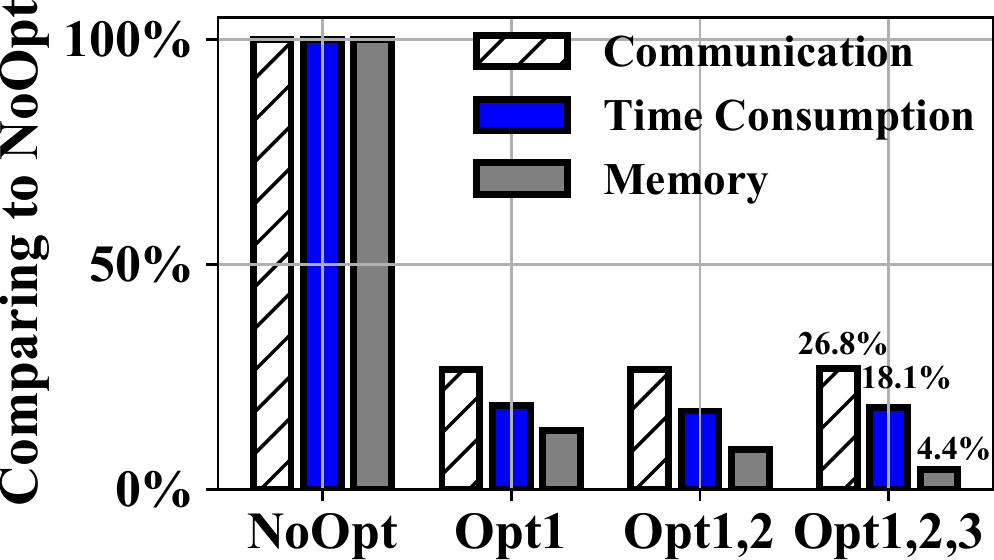}
    \vspace{+3mm}
    \caption{Effectiveness of the Proposed Optimizations.}
    \label{fig:optimizations}
  \end{minipage}
\end{figure*}

\parab{Time Consumption:} \Cref{fig:large_scale_svd_time_with_ppdsvd} shows the time consumption of HE-based SVD (\ie, PPDSVD) and \solution on large-scale data. Specifically, we use synthetic data matrix $\mathbf{X} \in \mathbb{R}^{m \times n}$, fix $m=1$K, and vary $n$ from 10 to 50 million. The experiment of PPDSVD stops at $n=2$K because it takes too much time to further increase $n$. PPDSVD takes 53.1 hours to factorize a 1K $\times$ 2K matrix which is $10000\times$ slower than \solution. Meanwhile, we also observe that the time consumption of PPDSVD increases quadratically with $n$ when fixing $m$, while \solution increase linearly. Approximately, PPDSVD needs more than 15 years to factorize a 1K $\times$ 100K matrix, \ie, million-scale elements. \solution only needs 16.3 hours to factorize a 1K $\times$ 50M matrix, which contains 50 billion elements.

\parab{Communication:} \solution also has more than 10 times smaller communication size compared with PPDSVD, which is presented in \Cref{fig:small_scale_svd_comm}. \Cref{fig:time_to_bandwidth_svd} and \Cref{fig:time_to_rtt_svd} show the efficiency when we change networking bandwidth and latency, and \solution works well given different networking conditions. Figure \ref{fig:client_to_communication_svd} shows the amount of communication data per user when we change the data size of each user (\ie, $n_i$) and the number of users. Each user's communication size linearly increases with the size of local data.

\parab{Hyper-parameter (Block Size):} Block size is the only hyper-parameter in our solution and we present the system efficiency using different block size in \Cref{fig:block_size}. \solution's time consumption slowly increases with $b$. We suggest using a proper block size to gain enough privacy protection, which is discussed in \S\ref{sec:attack}, and benefit from the efficiency brought by the block-based optimizations.

\subsection{Evaluation on the Applications} \label{sec:evaluation_app}

In this section, we evaluate \solution on three applications: PCA, LR, and LSA regarding accuracy and efficiency.

\parab{Lossless:} The lossless evaluations of three applications are presented in \Cref{tab:lossless_evaluation}. For PCA and LSA, we measure the precision by calculating the the projection distance \cite{FedPCA} (\ie, $||\mathbf{U}\mathbf{U}^T-\hat{\mathbf{U}}\hat{\mathbf{U}}^T||_2$) to standalone SVD. For LR, we report the mean square error (MSE) on the training data. PCA and LSA share the same evaluation results because their nature are both truncated SVD. In \Cref{tab:lossless_evaluation}, we set $r=10$ for PCA and LSA. Compared with FedPCA and WDA, \solution consistently has more than 10 orders of magnitude lower projection distance on PCA and LSA applications. On LR application, \solution has the lowest MSE compared to FATE and SecureML which solves LR using SGD. Moreover, \solution only needs to factorize the data once to find the optimal solution, while SGD-based method usually needs multiple epochs of training to converge.

\parab{Efficiency:} \Cref{fig:large_scale_lr} shows the LR time consumption of \solution, FATE and SecureML when we fix $n=1$K and vary $m$ from 1M to 50M, and the results show that \solution is 100x faster than SecureML and 10x faster than FATE. \Cref{fig:time_to_bandwidth_lr} and \Cref{fig:time_to_rtt_lr} show the time consumption of LR under different network bandwidth and latency, the results show that \solution is less sensitive to network compared with SecureML, and \solution achieves consistently best performance under different network conditions. We have performed billion-scale data evaluation on all the applications and the results are reported in \Cref{tab:large_scale_application}. The results show that \solution is practical and successfully supports billion-scale applications.

\begin{table}[h]
    \centering
    \small
    \renewcommand\arraystretch{0.9}
    \setlength{\tabcolsep}{0.25em}
	\caption{Evaluate Applications on Billion-Scale Data. The data is uniformly partitioned on 2 users and network bandwidth=1Gb/s, RTT=50ms.}
	\begin{tabular}{|c|c|c|c|}
		\hline
		Application & Datasets & Data Size & Time \\
		\hline
		\tabincell{c}{PCA\\(top-$r=5$)} & \tabincell{c}{Synthetic Data} & \tabincell{c}{100K $\times$ 1M \\ (100 Billion Elements)} & 32.3 Hours \\
		\hline
		\tabincell{c}{LSA\\(top-$r=256$)} & \tabincell{c}{MovieLens-25M \\ (RealWorld)} & \tabincell{c}{62K $\times$ 162k \\ (10 Billion Elements)} & 3.71 Hours \\
		\hline
		LR & \tabincell{c}{Synthetic Data} & \tabincell{c}{1K $\times$ 50M \\ (50 Billion Elements)} & 13.5 Hours \\
		\hline
	\end{tabular}
	\label{tab:large_scale_application}
	\vspace{-4mm}
\end{table}

\vspace{1mm}
\subsection{Attacks} \label{sec:attack}

We have provided privacy analysis of \solution in \S\ref{sec:privacy_proof} showing that CSP cannot recover the raw data from the masked data when having no prior knowledge. In this section, we assume the CSP empirically choose data distributions as prior knowledge and perform independent component analysis (ICA) attacks \cite{ica_attack} on the masked data. Meanwhile, we set block size to different values and observe its impact on the effectiveness of privacy protection.

The ICA attack is the SOTA attack method on masked data proposed by \citet{ica_attack} for revealing raw data from masked databases. The main idea is to treat the masked data as a linear combination of different data sources, which are assumed to be independent and non-gaussian distributed. The attackers empirically choose distributions of the data sources (\eg, using sigmoid as the cumulative probability distribution function), and try to find the inverse of the linear combination that maximizes the likelihood function.

\vspace{-2mm}
\begin{table}[h]
	\centering
	\small
	\renewcommand\arraystretch{0.9}
	\setlength{\tabcolsep}{0.5em}
	\caption{ICA attacks on the masked data. Pearson correlation between the attack results and raw data are reported.}
	\begin{tabular}{|c|c|c|c|c|}
		\hline
		Attacks & $b$ & MNIST & ML-100K & Wine \\
		\hline
		Random Values & NA & 0.12590 & 0.17957 & 0.49313 \\
		\hline
		ICA & 10 & 0.20329 & 0.18623 & 0.44268 \\
		ICA($b$) & 10 & 0.32029 & 0.28434 & 0.45971 \\
		\hline
		ICA & 100 & 0.12590 & 0.18387 & 0.45183 \\
		ICA($b$) & 100 & 0.13051 & 0.20910 & 0.45826 \\
		\hline
		ICA & 1000 & 0.11104 & 0.18020 & 0.44712 \\
		ICA($b$) & 1000 & 0.12531 & 0.18057 & 0.44862 \\
		\hline
	\end{tabular}
	\vspace{-3mm}
	\label{tab:ica_attacks}
\end{table}

In our experiments, we run ICA attack on both side of the masked data since \solution has two masks, and \Cref{tab:ica_attacks} shows the results. Meanwhile, we also perform attacks assuming the CSP knows the block size $b$, denoted as ICA($b$) in \Cref{tab:ica_attacks}, which reduces the number of parameters to solve in the attack. We use Pearson correlation to assess the attack results. Since ICA has disordered outputs (\ie, recovered data might be shuffled by row or by column), we compute n-to-n matching Pearson correlation between the attack results and real data, and report the maximum value. We use random value as the baseline, and if the Pearson correlation between the attack results and the raw data is close to the Pearson correlation between random value and raw data, then we can conclude that the attack fails in recovering valid information. We can observe from \Cref{tab:ica_attacks} that 1) ICA($b$) is more effective than ICA, which means that knowing $b$ is helpful to the attacks; 2) When increasing $b$ from 10 to 1000, attacking effectiveness of both ICA and ICA($b$) decrease; 3) When setting $b=1000$, all the attacks fail in recovering valid information.

In conclusion, 1) The Pearson correlation between the attack results and raw data decreases with the increase of block size, when the block size is large enough (\eg, 1000 in our experiments), the ICA attack fails in recovering valid information; 2) Leaking the block size reduces the complexity of ICA attack, however, the attack still could be defensed as long as the block size is large enough.

In the application, since different datasets have various distributions, we suggest the users run local ICA attacks and choose a proper block size that can resist the attack.

\subsection{Effectiveness of Proposed Optimizations} \label{sec:optimization}

In this section, we compare the efficiency with and without the proposed optimizations to show the effectiveness of our design.

We categorize three types of optimizations from our system: 1) Opt1: the block-based optimizations including efficient mask initialization, data masking, and recovery of $\mathbf{V'^T}$; 2) Opt2: mini-batch secure aggregation; 3) Opt3: advanced disk offloading. \Cref{fig:optimizations} shows the evaluation results using 10K $\times$ 50K synthetic data. Compared with using no optimizations, our solution reduces the communication, time consumption, and memory usage by 73.2\%, 81.9\%, and 95.6\%, respectively. To further demonstrate the effectiveness of Opt3, we compare the efficiency of RAM+AdvancedOffLoading and RAM+SwapOffLoading on larger data (10K $\times$ 100K), the results show that our solution reduces the time consumption by 44.7\% compared with swap disk offloading scheduled by OS.

\section{Related Work}

Apart from the federated SVD methods introduced in \S\ref{sec:introduction}, there are also other research topics that closely related to our work:

\parab{Privacy-preserving Funk-SVD:} The Funk-SVD is utilized in the federated recommender system \cite{DBLP:series/lncs/YangTZCY20}. The major difference between Funk-SVD and SVD is that Funk-SVD runs on the sparse rating matrix. \citet{chai2020secure} solved the federated Funk-SVD problem using HE. \citet{berlioz2015applying} proposed a DP-based Funk-SVD method.

\parab{Outsourcing matrix factorization techniques:} The secure outsourcing computation is a traditional research topic. \citet{zhang2019secure} proposed a secure outsourcing computation framework for PCA-based face recognition. \citet{duan2019secure} proposed outsourcing computation frameworks for non-negative matrix factorization. \citet{luo2017secfact} proposed a masking based outsourcing computation method for QR and LU factorization.

\section{Conclusion}

In this paper, we propose a practical lossless federated SVD method over billion-scale data. Compared with the existing federated SVD methods, \solution is lossless and efficient. The experiments show that \solution is over $10000\times$ faster than HE-based method and has 10 orders of magnitude smaller error compared DP-based method.

\begin{acks}
The work is supported by the Key-Area Research and Development Program of Guangdong Province (2021B0101400001), the NSFC Grant no. 61972008, the Hong Kong RGC TRS T41-603/20R, the National Key Research and Development Program of China under Grant No.2018AAA0101100, and the Turing AI Computing Cloud (TACC) \cite{xu2021tacc}.
\end{acks}

\bibliographystyle{ACM-Reference-Format}
\bibliography{ref.bib}

\appendix

\section{Datasets and Baseline Models} \label{appendix:dataset}

\parab{Datasets:} We have used five datasets in the experiments. Following is the detailed description and the parameter settings:

\begin{icompact}
	\item MNIST \cite{lecun1998gradient}: A standard hand-written digits image testset, and each image contains 784 (i.e., $28 \times 28$) features. We take 10K labeled images in the experiment, thus $X_{mnist} \in \mathbb{R}^{784 \times 10K}$.
	\item Wine \cite{Dua:2019}: The physicochemical data for 6498 variants of red and white wine, and each sample has 12 features. Thus $X_{wine} \in \mathbb{R}^{12 \times 6489}$.
	\item movielens \cite{ml100k}: Movielens dataset describes people’s expressed preferences for movies. It contains contains millions of users' rating records over different movies. We select two groups of movielens data: movielens-100K and movielens-25M for our experiment. Movielens-100k contians 943 users' rating on 1682 movies, thus $X \in \mathbb{R}^{1682 \times 943}$. Movielens-25M contians 162542 users' rating on 59047 movies, thus $X \in \mathbb{R}^{59047 \times 162542}$.
	\item Synthetic data \cite{FedPCA}: Apart from the real-world datasets, we also use synthetic data in the evaluation. The synthetic data is generated from a power-law spectrum $Y_{\alpha} \sim Synth(\alpha)^{m \times n}$ using $\alpha=0.01$. More specifically, $Y=U\Sigma V^T$, where $[U, \sim]=QR(N^{m \times m}),[V, \sim]=QR(N^{m \times n}), \Sigma_{i,i}=i^{-\alpha}$, and $N^{m \times n}$ is an matrix with i.i.d entries drawn from $\mathcal{N}(0,1)$.
\end{icompact}

\parab{Baseline Models:} We compare FedSVD with three existing works, and following is the detailed introduction and parameter setting.

\begin{icompact}
    \item WDA-PCA \cite{WDA}: In the weighted distributed averaging PCA (WDA-PCA), the participants upload local rank-$k$ approximation of the covariance matrix to the server, which will aggregate all the approximations through weighted average and do a rank-$k$ PCA on the aggregated matrix to get the final results. WDA-PCA reduces the private data leakage since each users only uploads a rank-$k$ approximation of the covariance matrix. In our experiments, we only compared FedSVD and WDA-PCA in PCA applications, since WDA-PCA is specially designed for rank-$k$ PCA and not suitable for SVD tasks.
    \item FedPCA \cite{FedPCA}: Federated principal component (FedPCA) analysis is a federated, asynchronous, and $(\epsilon,\delta)$-differentially private algorithm. Follow the setting in \cite{FedPCA}, we set $\epsilon=0.1,\delta=0.1$. We compare FedSVD and FedPCA in both PCA and SVD tasks. 
    \item PPD-SVD \cite{PPD-SVD}: Privacy-preserving decentralized SVD (PPD-SVD) used homomorphic encryption to protect user's private data during the joint computation of covariance matrix, then decrypt the covariance matrix and do regular SVD tasks. According to the original paper's setting, we set the key size of HE to 1024.
\end{icompact}

\parab{Hardware:} All the experiments are performed on a Ubuntu 20.04 Server with a 3.6GHz 8-core CPU, 128GB RAM, and 2TB SSD. The programming language is Python. For all the experiments, we put participants into different Docker containers, which are connected using the docker-bridge network, and we simulate the network bandwidth and latency between containers.

\end{document}